\documentclass{article}
\usepackage{graphicx}
\usepackage[numbers,sort&compress]{natbib}
\usepackage{amsmath,amsthm}
\usepackage{amssymb}
\usepackage{authblk}
\newtheorem{theorem}{Theorem}
\def\Im{{\rm Im}}
\def\Re{{\rm Re}}
\date{ }

\title{Construction of  complex potentials for multiply connected domains}
\author[1]{P.~N.~Ivanshin}
\affil[1]{Lobachevskiy Institute of Mathematics \& Mechanics, Kazan Federal
University, Kremlevskaya st., 35, Kazan, 420008, Russia.
E-mail: pivanshi@yandex.ru}

\begin{document}
\maketitle

\begin{abstract}

The method of reduction of a Fredholm integral equation to the linear
system is generalized to construction of a complex potential --- an analytic
function in an infinite multiply connected domain with a simple pole at
infinity which maps the domain onto a plane with horizontal slits. We
consider a locally sourceless, locally irrotational
flow on an arbitrary given $n$-connected infinite domain with impermeable boundary. The complex potential has the form of a Cauchy integral with one linear and $n$ logarithmic summands. The method is easily computable.


\end{abstract}

\section*{ Introduction}

Conformal mappings by the analytical functions of complex variable play an
important part in solution of many problems of mechanics and mathematics,
particularly in the case of plane potential fields and Laplace equation solution \cite{1}. It is known that a complex potential of a plane locally sourceless, locally irrotational steady flow in any 
$n$-connected infinite domain may have only logarithmic summands to the analytic function with the given simple pole at infinity \cite{BPR}. Here we present the method of construction of a complex potential for this flow in any multiply connected domain with impermeable boundary through approximate solution of the Fredholm equation. We give the following formulation of the problem: given the velocity at the infinity and the circulations around
the contours we are able to find the analytic function with constant imaginary parts 
at the domain boundary and the given simple pole at infinity. We prove convergence of the method and give certain
clarifying examples of its application. We base our construction of the complex
potential in an infinite multiply-connected domain on the approximation of the analytic mapping of our domain on the domain with horizontal slits with the given simple pole at   infinity.

The solution of the integral equations in our method is reduced to the
solution of an infinite linear system. We obtain an approximate mapping
by solving the finite system with a truncated matrix. 

  The method computational complexity equals $O(M^3)$, here $M$ is the order of the corresponding system. 
%  
%  This form allows us to apply the 
%differentiation techniques for solutions of some physical problems in
%corresponding domains. 

 % We  note that
%our method is applicable only for the domains with smooth boundary
%components or the boundary components with vertices of the angles exceeding
% $\pi$.

\section*{Complex potential construction  for the flow around  several contours by means of Cauchy integral}

Consider an infinite $n$-connected domain $D_z$ bounded by the simple smooth curves $L_s$  given by the equations
$$
L_s = \{z=z_s(t),\  z_s(0)=z_s(2\pi),\  t\in [0,2\pi]\},\  s =1, \ldots, n.
$$

We also assume that the boundary curves $L_s$ complex representations are as
follows:
$$
z_s(t)= \sum\limits_{k=-m_s}^{n_s}d_{k s} e^{i k t},\ t\in [0,2\pi], \ s=1, \ldots, n.
$$
The parametrization traces  the contours $L_s$, $s = 1, \ldots, n$, counterclockwise.

Note that the flow may include circulations around certain connected components of the domain. %So we assume that we know  the velocity   at the horizontal slits or the circulation aroud them.

%Definition 1. We call the unit  disk with $m$  circular slits $\zeta = R_s e^{i \theta}$, $\theta_{1 s} <\theta <\theta_{2 s}$,  $0<R_s<1$, $\theta_{2 s} - \theta_{1 s} <2\pi$, $s=1,\ldots,m$,  and with $n-m$ radial slits $\zeta = R e^{i \theta_j}$, $0<R_{1 j} <R <R_{2 j}<1$,  $0<\theta_j<2\pi$,  $j=m+1,\ldots,n$, an $(n+1)-$  connected canonical domain of the first type.
%
%Definition 2. We call the annulus with the exterior radius 1,  with  the interior radius $r$, $r<1$, and with  $(m-1)$  circular slits $\zeta = R_s e^{i \theta}$, $\theta_{1 s} <\theta <\theta_{2s}$,  $r<R_s<1$, $\theta_{2 s} - \theta_{1 s} <2\pi$, $s=1,\ldots,m-1$, and with $n-m$ radial slits $\zeta = R e^{i \theta_j}$, $0<R_{1 j} <R <R_{2 j}<1$,  $0<\theta_j<2\pi$,  $j=m+1,\ldots,n$, an $(n+1)-$  connected canonical domain of  the second type.

\begin{theorem}  

A complex potential  $f(z)$ of the flow in any $n$-connected domain $D_{z}$ can be approximately constucted by reduction to solution of a linear  system. %The map is unique under the  following conditions: $f(0)=0$. 
\end{theorem}

\begin{proof}   

Assume that $v \in \mathbb{C}$ is the velocity of the flow at  infinity.
Consider the complex potential in the form of the function $\phi(z)=\overline{v} z+\sum\limits_{s=1}^{n} c_s \log (z-z^*_{s})+\psi(z)$ \cite{BPR}. Here $z^*_{s}$ is a point inside the domain bounded by the contour $L_s$, $c_s=\frac{\Gamma_s}{2 \pi i}$, $\Gamma_s$ is the circulation around the  contour $L_s$, $s=1, \ldots, n$,  \cite{BPR},  and $\psi(z)$ is the unknown analytical in $D_z$ function. %with the known values at the contours $L_s$, $s=1, \ldots, n$.

Any contour $L_s$, $s=1, \ldots, n$, is  a part of a flow line for this complex potential, so the imaginary part of the function $\phi(z)$ is constant on these contours.
%We assume that $0\in D_z$ and $A+iB=0$ without loss of generality. We give the constructive proof. We construct the complex potential of the flow by reparametrization of the given boundary representations. So we search for the function $t_0(\theta),$ $\theta \in [0, 2\pi]$,  for  the functions $t_s(\theta),$  $s = 1, \ldots, m$,   $\theta \in [\theta_{1 s},\theta_{2 s}]$, and for the functions $t_j(R)$, $R\in [R_{1 j},R_{2 j}]$,   such that the values $z_s(t_s(\theta))$, $s = 0, \ldots, m$,  $z_j(t_j(R))$, be the boundary values of an analytic function in the domain. The parameters $R_s$, $\theta_{1 s}$, $\theta_{2 s}$, $s=1,\ldots,m$, $R_{1 j}$, $R_{2 j}$,  $\theta_j$,  $j=m+1,\ldots,n$, are also unknown and will be found within the solution process.

%Let us  consider the analytic in the domain $D_z$ function $\zeta(z)$ which maps conformally the
%domain $D_z$ onto $D_{\zeta}$ with the correspondence $\zeta(0) = 0$  and the 
%analytic in $D_z$ function $\log \frac{z}{\zeta}$.

%We consider either the flow without circulation with $f(z)=\overline{v} z$ as the right-hand side of our equation on $q(z)$ or the flow with circulation for which we put the pole with the logarythmic summands of type $f(z)=\sum\limits_{s=1}^{n} c_s \log (z-z^*_{s})$, $j=1, \ldots, n$, into the domains. 
  According to \cite{6} the necessary and sufficient
condition for  $\psi(t)$, to be analytic in $D_z$ are the boundary relations
\begin{equation}\label{1}
\psi(z_s(t)) =\sum\limits_{\sigma=1}^n \frac{1}{\pi i} \int\limits_0^{2\pi}
\psi(z_{\sigma}(\tau))  [\log(z_{\sigma}(\tau)-z_s(t)) ]'_\tau d\tau.
\end{equation}
%$$
%\sum\limits_{j=m+1}^{n} \frac{1}{\pi i} \int\limits_0^{2\pi }
%\log \frac{z_j(\tau)}{R_j(\tau) e^{i \theta_j}} [\log(z_j(\tau)-z_s(t)) ]'_\tau d\tau,  
%$$
where $ t\in[0, 2\pi]$, $ s=1, \ldots, n$.%,  $R_0=1$, and 
%$$
%\log \frac{z_j(t)}{R_j(t) e^{i \theta_j}} =\sum\limits_{s=0}^m \frac{1}{\pi i} \int\limits_0^{2\pi}
%\log \frac{z_s(\tau)}{R_s e^{i \theta_s(\tau)}} [\log(z_s(\tau)-z_j(t)) ]'_\tau d\tau +, 
%$$
%$$
%\sum\limits_{k=m+1}^{n} \frac{1}{\pi i} \int\limits_0^{2\pi}
%\log \frac{z_k(\tau)}{R_k(\tau) e^{i \theta_k}} [\log(z_k(\tau)-z_j(t)) ]'_\tau d\tau,  \eqno(2)
%$$
%where $ t\in[0, 2\pi]$ , $j=m+1, \ldots, n$.

We introduce the new functions $p_s(t)$, $q_s(t)$ as follows: $\psi(z_s(t))=p_s(t)+iq_s(t)$,  $s =1, \ldots, n$. %, and $p_j(t)=\log |z_j(t)| -\log R_j(t)$, where $R_j(t)$ is the radius of the image of the point of $z_j(t)$, $j =m+1, \ldots, n$. 
Note that 
\begin{equation}\label{2}
q_s(t)=-\Im[\overline{v}z_s(t)+ \sum\limits_{\sigma=1}^{n} c_{\sigma} \log (z_s(t)-z^*_{\sigma})]+C_s, 
\end{equation}
 here $C_s$ is constant for any $s=1, \ldots, n$.

We separate the real part of both sides of equation (\ref{1}):
\begin{eqnarray}\label{3}
p_s(t)=\sum\limits_{\sigma=1}^n \frac{1}{\pi} \int\limits_{0}^{2\pi} p_{\sigma}(\tau)[\arg(z_{\sigma}(\tau)-z_s(t)) ]'_\tau  d\tau +\\
 +\sum\limits_{\sigma=1}^n \frac{1}{\pi}
\int\limits_{0}^{2\pi} q_{\sigma}(\tau)[\log|z_{\sigma}(\tau)-z_s(t)| ]'_\tau  d\tau . 
\end{eqnarray}

%$$
%+\sum\limits_{j=m+1}^n \frac{1}{\pi} \int\limits_{0}^{2\pi} [\arg z_j(\tau)-\theta_j] [\arg(z_j(\tau)-z_s(t)) ]'_\tau  d\tau -
%$$
%$$
% - \sum\limits_{j=m+1}^n \frac{1}{\pi}
%\int\limits_{0}^{2\pi} p_j(\tau) [\log|z_j(\tau)-z_s(t)| ]'_\tau  d\tau,  \quad s=0, \ldots, m.
%$$
%We separate the real part of both sides of equation (2):
%$$
%p_j(t)=
%\sum\limits_{s=0}^m \frac{1}{\pi} \int\limits_{0}^{2\pi} [\log |z_s(\tau)|-\log R_s] [\arg(z_s(\tau)-z_j(t)) ]'_\tau  d\tau +
%$$
%$$
%+ \sum\limits_{s=0}^m \frac{1}{\pi}
%\int\limits_{0}^{2\pi} q_s (\tau) [\log|z_s(\tau)-z_j(t)| ]'_\tau  d\tau +
%$$
%$$
%+\sum\limits_{k=m+1}^n \frac{1}{\pi} \int\limits_{0}^{2\pi} p_k(\tau) [\arg(z_k(\tau)-z_j(t)) ]'_\tau  d\tau +
%$$
%$$
%+ \sum\limits_{k=m+1}^n \frac{1}{\pi}
%\int\limits_{0}^{2\pi} [\arg[z_k(\tau)-\theta_k] [\log|z_k(\tau)-z_j(t)| ]'_\tau  d\tau,  \quad j=m+1, \ldots, n. \eqno(4)
%$$

After differentiating  relation  (\ref{2}) with respect to $t$ and integrating the results by parts, we
obtain the following relations on the functions $p_s'(t)$:
\begin{equation}\label{4}
p_s'(t)=\sum\limits_{\sigma=1}^n\frac{1}{\pi} \int\limits_{0}^{2\pi} p_{\sigma}'(\tau) K_{\sigma, s}(\tau,t)  d\tau+ Q_s(t), 
\end{equation}
%$$
%+\sum\limits_{j=m+1}^n\frac{1}{\pi} \int\limits_{0}^{2\pi} p_j'(\tau) L_{j, s}(\tau,t)  d\tau , \ \ s=0, \ldots, m, \eqno(5)
%$$
%$$
%p_j'(t)=-\sum\limits_{s=0}^m\frac{1}{\pi} \int\limits_{0}^{2\pi} q_{s}'(\tau) L_{s, j}(\tau,t)  d\tau+
%$$
%$$
%+\sum\limits_{k=m+1}^n\frac{1}{\pi} \int\limits_{0}^{2\pi} p_k'(\tau)K_{k,j}(\tau,t)  d\tau + P_j(t), \ \ j=m+1, \ldots, n, \eqno(6)
%$$
where
$$
 K_{\sigma s}(\tau,t) = - [\arg(z_{\sigma}(\tau)-z_s(t)) ]'_t,\ \ L_{j,s}(\tau,t)=[\log|z_j(\tau)-z_s(t)| ]'_t,
 $$
 $$ 
Q_s(t)=  \sum\limits_{\sigma=1}^n \frac{1}{\pi}
\int\limits_{0}^{2\pi} q'_{\sigma}(\tau) L_{\sigma,s}(\tau,t) d\tau, 
$$
here by formula (2) $q'_{s}(t)=-\Im[\overline{v}z_s'(t)+ \sum\limits_{\sigma=1}^{n} \frac{c_{\sigma}z_s'(t)}{ z_s(t)-z^*_{\sigma}}]$, $s=1, \ldots, n$.
%$$ 
%P_j(t)=  \sum\limits_{s=0}^m \frac{1}{\pi}
%\int\limits_{0}^{2\pi} [\log |z_{s}(\tau)|]' K_{s,j}(\tau,t) d\tau-\sum\limits_{k=m+1}^n\frac{1}{\pi} \int\limits_{0}^{2\pi}(\arg z_k(\tau))' L_{k, j}(\tau,t)  d\tau. 
%$$

The kernel $L_{\sigma,s}$  has a singularity in the form of $\cot \frac{\tau-t}{2}$
for $\sigma = s$:
$$
(\log|z_s(\tau)-z_s(t)| )'_t =\Re \left(\log \sum\limits_{k=-m_s}^{n_s} d_{k s}[e^{i k\tau}-e^{i kt}]\right)'_t=\Re\left(\log \sin\frac{\tau-t}{2}+\right.
$$
$$
\left.+\log\left[ \sum\limits_{k=1}^{n_s} d_{k s}e^{ikt}  \sum\limits_{l=0}^{k-1} e^{i l (\tau-t)}- \sum\limits_{k=1}^{m_s} d_{(-k) s}e^{-ik\tau}  \sum\limits_{l=1}^{k-1} e^{i l (\tau-t)}\right]\right)_t'=
$$
$$
=-\frac{1}{2}\cot \frac{\tau - t}{2} + \left(\log\left|\sum\limits_{k=1}^{n_s} d_{k s}e^{ikt}  \sum\limits_{l=0}^{k-1} e^{i l (\tau-t)}- \sum\limits_{k=1}^{m_s} d_{(-k) s}e^{-ik\tau}  \sum\limits_{l=1}^{k-1} e^{i l (\tau-t)}\right|\right)_t'.
$$

The Cauchy principal value integral
$$
\frac{1}{\pi}
\int\limits_{0}^{2\pi} [\log |z_{\sigma}(\tau)|]' \cot\frac{\tau-t}{2} d\tau
$$
can be calculated via Hilbert formula \cite{9} as in \cite{8}.

Finally we obtain the following system of Fredholm integral equations of the second kind which can be written in the operator form as follows:
$$
\left(\begin{array}{ccccccc}
I- K_{1,1}& - K_{2,1} & \ldots &- K_{n,1} \\
- K_{1,2} & I- K_{2,2} & \ldots &- K_{n,2} \\
\ldots & \ldots & \ldots & \ldots \\
- K_{1,n} &- K_{2,n} & \ldots & I- K_{n,n} 
\end{array}\right) \left(\begin{array}{c}p_1'\\
\vdots \\
p_n'\end{array}\right) =
$$
\begin{eqnarray} \label{eqb}
=\left(\begin{array}{c}Q_{1}  \\
\vdots\\
Q_{n}\end{array}\right).
\end{eqnarray}

The last operator system can be reduced to the infinite linear system over the Fourier coefficients of the unknown functions $p_s'(t)$, $s =1, \ldots, n$, if we find the coefficients of double Fourier expansions of  the kernels of integral operators and compare the coefficients with the same trigonometric functions.   Approximate solution of the infinite system over Fourier coefficients of the unknown functions is a solution of a truncated system over the Fourier coefficients of the unknown functions.

Existence of the exact solution of system (\ref{4})  and convergence of the approximate
solution to the exact one provided $M \to \infty$ were proved in \cite{7} for the case of conformal mapping of a simply connected domain. This proof can be
applied to the case of multiply connected domain if we replace the corresponding
space $l^2$ by the space $ l^2 \times l^2 \times \cdots \times l^2$.

We search for the approximate solution of system (\ref{4}) in the form of Fourier polynomials:
\begin{equation}\label{5}
p_s'(t)= \sum\limits_{l=1}^{M} \alpha_{l s} \cos lt+ \beta_{l s} \sin lt, \  \ s=1, \ldots, n. 
\end{equation}
%$$
%p_j'(t)= \sum\limits_{l=1}^{M} \alpha_{l j} \cos lt+ \beta_{l j} \sin lt, \  \ j=m+1, \ldots, n, \ \ t\in[0,2\pi].    \eqno(7)
%$$

Now integral Fredholm equations of the second kind in (\ref{4})  can be reduced to
the linear system over Fourier coefficients $\alpha_{l s}$ and $\beta_{l s}$, $s=1, \ldots, n$:
%, $\alpha_{l j}$ and $\beta_{l j}$, $ j=m+1, \ldots, n$

$$
\left( \begin{matrix}A_{11} & B_{11}& A_{12}& \ldots& B_{1 n}\\ 
\vdots &  \vdots &  \vdots &  \ddots & \vdots  \\
C_{n 1} & D_{n 1}& C_{n 2}&\ldots& D_{n n} \end{matrix} \right) \times
$$
$$
\times \left( \begin{matrix} \alpha_1 \\ \beta_1\\  \alpha_2 \\ \beta_2\\ \vdots\\ \beta_n \end{matrix} \right) =
\left( \begin{matrix} a_1 \\ b_1\\  a_2 \\ b_2\\ \vdots \\ b_n \end{matrix} \right),
$$
where $\alpha_s =(\alpha_{1 s},\ldots,\alpha_{n s})^T$, $\beta_s =(\beta_{1 s},\ldots,\beta_{n s})^T$. The vectors   $a_s =(a_{1 s},\ldots,a_{n s})^T$, $b_s =(b_{1 s},\ldots,b_{n s})^T$ on the right-hand side of the system consist of the elements
$$
a_{j s}=\frac{1}{\pi}
\int\limits_{0}^{2\pi} Q_s(t)\cos jt dt,\  b_{j s}=\frac{1}{\pi}
\int\limits_{0}^{2\pi} Q_s(t)\sin jt dt,\  \ j=1, \ldots, n,\ s=1, \ldots, n,
$$
%$$
%a_{j k}=\frac{1}{\pi}
%\int\limits_{0}^{2\pi} P_k(t)\cos jt dt,\  b_{j k}=\frac{1}{\pi}
%\int\limits_{0}^{2\pi} P_k(t)\sin jt dt,\  \ j=1, \ldots, m,\ k=m+1, \ldots, n,
%$$

The block matrices $A_{\sigma s}$, $B_{\sigma s}$, $C_{\sigma s}$, $D_{\sigma s}$,   $\sigma, s=1, \ldots, n$,  of size $ M\times M$ consist of the elements
$$
A_{\sigma s j k} = \delta_{\sigma s} \delta_{j k} - \frac{1}{\pi^2}\int\limits_{0}^{2\pi}\cos k\tau d\tau  \int\limits_{0}^{2\pi}K_{\sigma s}(\tau, t) \cos jt dt, 
$$
$$
B_{\sigma s j k} =  - \frac{1}{\pi^2}\int\limits_{0}^{2\pi}\sin k\tau d\tau  \int\limits_{0}^{2\pi}K_{\sigma s}(\tau, t) \cos jt dt, 
$$
$$
C_{\sigma s j k} =  - \frac{1}{\pi^2}\int\limits_{0}^{2\pi}\cos k\tau d\tau  \int\limits_{0}^{2\pi}K_{\sigma s}(\tau, t) \sin jt dt, 
$$
$$
D_{\sigma s j k} = \delta_{\sigma s} \delta_{j k} - \frac{1}{\pi^2}\int\limits_{0}^{2\pi}\sin k\tau d\tau  \int\limits_{0}^{2\pi}K_{\sigma s}(\tau, t) \sin jt dt,
$$
where $ j, k=1, \ldots, n$, $\delta_{r t}$  is the Kronecker delta function.

 Now this  system can be reduced to the infinite equation system over the unknown coefficients $\alpha_{k,j}$, $\beta_{k,j}$, $j=0,\ldots,n$, $k=1,\ldots,\infty$, in the equivalent  form 

 \begin{eqnarray}\label{11steq}
\tilde{Y}=\tilde{P} \tilde{Y}+\tilde{Q},
\end{eqnarray}
  where $\tilde{Y}=(\alpha_{1,0},\alpha_{1,1},\ldots,\alpha_{1,n},\beta_{1,0},\beta_{1,1},\ldots,\beta_{1,n},\alpha_{2,0},\ldots)\in l_2$, the infinite matrix $\tilde{P}$ consists of the elements  $\frac{1}{\pi^2}\int\limits_0^{2\pi}f(m t) dt \int\limits_0^{2\pi}g(p \tau) (\arg[z_k(\tau)-z_j(t)])'_{\tau} d\tau$ or $\frac{1}{\pi^2}\int\limits_0^{2\pi}f(m t) dt \int\limits_0^{2\pi}g(p \tau) (\ln|z_k(\tau)-z_j(t)|)'_{\tau} d\tau$ with $f(x)$,  $g(x)$ equal to $\cos x$ or $\sin x$. These elements being the Fourier coefficients of double Fourier series of   $(\mathrm{arg}(z_k(\tau)-z_j(t)))'_{\tau}$ or $ \ln|z_k(\tau)-z_j(t)|'_{\tau}$,  $\tilde{Q}$ being the sequence of the corresponding Fourier coefficients of the functions $q_j(t)$, $j=1,\ldots,n$.  

We need to construct the approximate solution $\tilde{y}_j(t)$ of equation system  (\ref{11steq})
in the trigonometric polynomial form $\tilde{y}_j(t)=\sum\limits_{k=1}^{M} \alpha_{k,j} \cos k t+\beta_{k,j} \sin k t$ in order to apply truncated linear system as in \cite{15}. So we have to find the vector $\tilde{Y}_{F}$ with zero coordinates starting with the $(2 M (n+1) +1)$-th one  which approximates the infinite vector $\tilde{Y}$. Further on we identify the vector-function $Y$, the integral operator $P$, the vector-function $Q$ with the sequence $\tilde{Y}$, the infinite matrix $\tilde{P}$ and the sequence $\tilde{Q}$, respectively. 

Evidently the kernels $(\mathrm{arg}(z_j(\tau)-z_k(t)))'_{\tau}$ of system (\ref{eqb}) are infinitely differentiable for $k\not=j$. Due to Cauchy theorem we have  $\mathrm{arg}(z_j(\tau)-z_j(t))=$ $\mathrm{arg}z'(t+\theta(\tau-t))$ with $\theta \in (0,1)$, so this function is well-defined for  $t, \tau \in [0,2\pi]$ and $\lim\limits_{\tau\to t} \mathrm{arg}(z_j(\tau)-z_j(t))'_{\tau}=$ $k(t) |z'(t)|/2$ where $k(t)$ is the curvature of the boundary curve at the corresponding point. It can be easily verified that the kernel  $(\mathrm{arg}(z_j(\tau)-z_j(t)))'_{\tau}$  is  at least twice differentiable with respect to both variables. So the  double complex Fourier coefficients of $(\mathrm{arg}(z_k(\tau)-z_j(t)))'_{\tau}$  have the following estimates: $|c_{k,j, l, t}|<\frac{U}{|l|^2 |t|^2}$.

%$\mathrm{arg}(z_j(\tau)-z_k(t))$ of system (\ref{1steq})  and on the function $Q(t)$ yield the following estimates of the Fourier coefficients for $\frac{\partial\mathrm{arg}(z_k(\tau)-z_j(t))}{\partial \tau}$: $|c_{k,j, l, t}|<\frac{U}{|l|^r |t|^p}$ and  of $Q(t)$ Fourier coefficients:$ |\tilde{c}_{l,j}|<\frac{T}{|l|^2}$. We denote by $Y$ the vector of the function $Y(t)$ Fourier coefficients. 

For $F=(n+1) 2 M$  integral equation system (\ref{eqb})    reduces to infinite linear system  (\ref{11steq}) which can be presented as follows:
$$
\left( \begin{matrix}I_F-P_F & S \\ R & I_{\infty}-V \end{matrix} \right)
\left( \begin{matrix} Y_1 \\Y_2 \end{matrix} \right)  =
\left( \begin{matrix}Q_1  \\  Q_2 \end{matrix} \right).
$$

 Here  $P_F$ is an $F\times F$ block matrix $P=\left( \begin{matrix} P_{1,1} & \ldots & P_{1,n} \\
 \ldots & \ldots& \ldots \\
P_{n,1} & \ldots & P_{n,n} \end{matrix} \right)$, $M \times M$ matrices $P_{j,k}$ correspond to integral summands of (\ref{eqb}), $j,k=0, \ldots, n$,  $S$ is an  $F\times \infty$ matrix, $R$ is an $\infty\times F$ matrix, $V$ is an $\infty\times \infty$ matrix, $I_F$ and $I_{\infty}$ are the identity matrices of  relative sizes. Each of the vectors $Q_1$ and $Y_1$ has $F$ coordinates, the vectors $Q_2$ and $Y_2$ have the infinite number of coordinates. The Fourier coefficients of the smooth functions tend to zero as their  numbers tend to infinity, so  the coefficients of the matrices $S, R$ and $V$ together with the coordinates of $Q_2$ decrease rapidly as $F\to \infty$. Due to Theorem assumptions and the Fourier coefficients speed of convergence to zero  the matrix norm of $V$ and the vector norm of $Y_2$ tend to zero as $F\to \infty$.
 
Let us prove that there exists the number $T\in \mathbb{N}$ such that the matrix operator $I_F-P_F$ is invertible  $\forall F>T$ since the limit for $P_F$ integral operator $P$  is compact  and the operator $I-P$ is invertible due to the lemma assumption. Note that we do not distinguish a finitely dimensional vector and the Fourier polynomial with the corresponding finite set of coordinates in our proof.  Recall first that due to chapter VI, paragraph 1 of \cite{LS}  $\|P-P_F\|  \to 0$ if $F\to \infty$. The operator norm that we deal with here is the usual operator norm for the Hilbert space mappings. Let us assume that  $\forall T \in \mathbb{N}$ 
there exists $s_l>T$ such that the spectrum of $P_{s_l}$ contains $1$.   Then there exists an infinite sequence  $( {v}_{s_l})_{l \in \mathbb{N}} \subset L^2$ such that  $\| {v}_{s_l}\|=1$ and $P_{s_l}  {v}_{s_l}= {v}_{s_l}$.  Let us prove that  then there should exist at least one limit point for the sequence $\{ {v}_{s_l} \}_{l \in \mathbb{N}}$. Since the operator $P$ is compact there exist both a subsequence $\{ {v}_{s_{k_j}}\}_{j \in \mathbb{N}}$ and an element $ {w}_0 \in L^2$ so that $P  {v}_{s_{k_j}} \to  {w}_0$, $(j \to \infty)$. Then $\|P_{s_{k_j}}  {v}_{s_{k_j}} - {w}_0\|=\|P_{s_{k_j}}  {v}_{s_{k_j}}-P  {v}_{s_{k_j}}+P  {v}_{s_{k_j}} - {w}_0\| \leq \|P_{s_{k_j}}  {v}_{s_{k_j}}-P  {v}_{s_{k_j}}\|+\|P  {v}_{s_{k_j}} - {w}_0\| \leq \|P_{s_{k_j}} -P \| +\|P  {v}_{s_{k_j}} - {w}_0\|  \to 0$, $(j \to \infty)$. Thus $\| {v}_{s_{k_j}}-{w}_0\|=\|P_{s_{k_j}} {v}_{s_{k_j}}- {w}_0\| \to 0$, $(j \to \infty)$. Hence $ {v}_{s_{k_j}} \to{w}_0$, $(j \to \infty)$. Note that since  $\| {v}_{s_{k_j}}\|=1$, $\forall j \in \mathbb{N}$, the element $ {w}_0$ is nondegenerate. Let us show now that   the relation $P  {w}_{0}= {w}_{0}$  holds true. Indeed, we have $\|P {w}_0- {w}_0 \| =\|P {w}_0+P  {v}_{s_{k_j}}-P  {v}_{s_{k_j}}- {w}_0 \| \leq \|P\| \| {w}_0-{v}_{s_{k_j}}\|+\|P  {v}_{s_{k_j}}+P_{s_{k_j}}  {v}_{s_{k_j}}-P_{s_{k_j}}  {v}_{s_{k_j}}- {w}_0\| \leq \| P\| \| {w}_0-{v}_{s_{k_j}}\|+\|P_{s_{k_j}}  {v}_{s_{k_j}}-P  {v}_{s_{k_j}}\|+\| {w}_0-P_{s_{k_j}}  {v}_{s_{k_j}}\| \leq \|P\| \| {w}_0+ {v}_{s_{k_j}}\| +\|P -P_{s_{k_j}}\|+\| {w}_0- {v}_{s_{k_j}}\| \to 0$, $(j \to \infty)$.   Hence the spectrum of $P$ contains $1$. A contradiction with one of the assumptions.

We now take the number $F$ so that  $\|V\|<1$ and the matrix $I_F-P_F$ possesses the inverse one.    
               Now we have the relation
$$
(I_F-P_F)[I_F-(I_F-P_F)^{-1}S(I_{\infty}-V)^{-1}R]Y_1=Q_1-S(I_{\infty}-V)^{-1}Q_2.
$$

Obviously one can choose the value of $F$ so large that   $\|S(I_{\infty}-V)^{-1}R\|=O(1/F^2)\leq r$ where $r<1$ is arbitrary small.  Now we  estimate  the norm of the difference between the solution $x_1$  and the solution $\tilde{x}_1$ of the truncated system $(I_F-P_F)\tilde{x}_1=y_1$:
$$
\|Y_1-\tilde{Y}_1\| \leq \frac{1}{1-r} \|(I_F-P_F)^{-1}\| \|S(I_{\infty}-V)^{-1}\| \|Q_2\|+\frac{r}{1-r}  \|(I_F-P_F)^{-1}\| \|Q_1\|.
$$

Consider the first summand on the right-hand  side of the last inequality.  Recall the Jackson's inequality: if $ {\displaystyle f:[0,2\pi ]\to \mathbb{C} }$ is an $ r$ times differentiable periodic function such that
${\displaystyle |f^{(r)}(x)|\leq 1,\quad 0\leq x\leq 2\pi ,}$ then, for every positive integer $n$, there exists a trigonometric polynomial${\displaystyle T_{n-1}}$ of degree at most ${\displaystyle n-1}$ such that  $|f(x)-T_{n-1}(x)| \leq \frac{ C(r)}{n^r}$ for any $x \in [0, 2 \pi]$ where ${\displaystyle C(r)}$ depends only on $r$ \cite{Ah}. So the  vector norm of $y_2$ can be estimated by this inequality  by $K/F^2$.
The second summand also behaves no better than $O(1/F^2)$. So the error due to the series tail is  $O(1/F^2)$.

The functions $\psi_s(t)$, $s = 1, \ldots, n$, can be restored via the derivatives  $p'_s(t)$ given by formula (\ref{5}) and the functions $q_s(t)$ with
 arbitrary constant summands $q_{0 s}$
$$
\psi_s(t)=\tilde{p}_s(t)+i q_s(t)+q_{0 s}, \ \ \tilde{p}_s(t)=\sum\limits_{l=1}^{M} \frac{\alpha_{l s}}{l} \sin lt- \frac{\beta_{l s}}{l} \cos lt, \, s=1, \ldots, n.
$$
%$$
%p_j(t)=p_{0 j}+\tilde{p}_j(t), \ \ \tilde{p}_j(t)=\sum\limits_{l=1}^{M} \frac{\alpha_{l j}}{l} \sin lt- \frac{\beta_{l s}}{l} \cos lt, \  \ t\in[0,2\pi]. \eqno(8) 
%$$

We choose the constant summand   $q_{01}=0$. 

We obtain the values of the  other constant summands $q_{0 s}$, $s=2, \ldots, n$,    in the following way. We take $n-1$ points  $z_k^{*}$ in  $n-1$ finite components  bounded by $L_k$, $k=2, \ldots, n$. The functions $\psi_s(t)$, $s=1, \ldots, n$, are the boundary values of the analytical in $D_z$ function, so the Cauchy integral with the corresponding density along the boundary of $D_z$ vanishes at the points $z_k^{*}$, $k=2, \ldots, n$. Therefore we have the linear complex system 
$$
\sum\limits_{s=1}^n \frac{1}{2\pi i} \int\limits_0^{2\pi}
\left[\tilde{p}_s(\tau)+i q_s(\tau)+q_{0 s}\right] [\log(z_{s}(\tau)-z^{*}_{k}) ]'_\tau d\tau=0, \, k=2, \ldots, n.
$$
%$$
%\sum\limits_{k=m+1}^n \frac{1}{2\pi i} \int\limits_0^{2\pi}
%\left[- i \theta_{k} + i \arg z_{k}(\tau)+ p_{0 k}+ \tilde{p}_k(\tau)\right] [\log(z_k(\tau)-z^{*}_j) ]'_\tau d\tau
%,\   \ j=1, \ldots, n,
%$$
over the unknown complex numbers $q_{0 s}$, $s=2, \ldots, n$. %Here $F(z_{s}(\tau))$ is the integral of the function $f$ at the  boundary component $L_s$, $s=1, \ldots, n$.

%We restore the values of $\theta_{1 j}$ , $\theta_{2 j}$, $j=1, \ldots, m$,  after we have restored  $q_{0 j}$. Really
%$$
%\theta_{1 j}=\min\limits_{t\in[0,2\pi]} [\arg z_i(t)-\tilde{q}_j(t)]-q_{0 j}, \  \theta_{2 j}=\max\limits_{t\in[0,2\pi]} [\arg z_i(t)-\tilde{q}_j(t)]-q_{0 j}.
%$$
%We restore the values of $R_{1 j}$ , $R_{2 j}$, $j=m+1, \ldots, n$,  after we have restored  $p_{0 j}$. Really
%$$
%\log R_{1 j}=\min\limits_{t\in[0,2\pi]} [\log| z_j(t)|-\tilde{p}_j(t)]-p_{0 j}, \  \log R_{2 j}=\max\limits_{t\in[0,2\pi]} [\log |z_j(t)|-\tilde{p}_j(t)]-p_{0 j}.
%$$

%So all parameters of the circular domain of the first type $D_{\zeta}$  have been found.

Now we have the functions $\psi_s(t)$,  $s=1, \ldots, n$, $t\in [0,2\pi]$,  and therefore we can restore the complex potential. %$\log R_k(t)=\log |z_k(t)|-p_k(t)$, $k=m+1, \ldots, n$.
% Note that $\theta_0(t)$ grows monotonically when $t$ grows from $0$ to $2\pi$, $\theta_0(2\pi)-\theta_0(0)=2\pi$, while each of the functions $\theta_s(t)$, $s=1, \ldots, m$, $\log R_k(t)$, $k=m+1, \ldots, n$,  is $2\pi$- periodic  with one interval of increase and one interval of decrease. We can restore the inverse to $\theta_0(t)$  monotonically increasing function $t_0(\theta)$ and we can restore the single-valued functions $t_s^{\pm}(\theta)$, $\theta\in [\theta_{1 s}, \theta_{2 s}]$, $s=1, \ldots, m$, and $t_s^{\pm}(R)$, $R\in [R_{1 k}, R_{2 k}]$, $k=m+1, \ldots, n$ . 

The approximate analytical potential  now has the form of the Cauchy integral 
$$
\phi(z)=\sum\limits_{s=1}^n \frac{1}{2\pi i} \int\limits_{0}^{2 \pi i}
 \frac{(\tilde{p}_s(\theta)+i q_s(\theta)+q_{0 s})z'_s(\theta)}{ z_s(\theta)-z}  d\theta+\overline{v} z+\sum\limits_{s=1}^{n} c_s\log (z-z^*_{s}) .
$$
%$$ 
%+ \sum\limits_{j=m+1}^n \frac{1}{2\pi i} \int\limits_{R_{1 j}}^{R_{2 s}}
% \frac{[z_j(t^{+}_j(R))-z_j(t^{-}_j(R))]  e^{i \theta_j}}{R e^{i \theta_j}-\zeta}  d R .
%$$

%We can apply the Cauchy integral in the form 
%$$
%f(\zeta)=\sum\limits_{s=1}^n \frac{1}{2\pi} \int\limits_{0}^{2\pi}
% \frac{z_s(t) e^{i \theta_s(t)}\theta_s'(t)}{ e^{i \theta_s(t)}-\zeta}  dt.
%$$
%$$ 
%+\sum\limits_{k=m+1}^n \frac{1}{2\pi i} \int\limits_{0}^{2\pi}
% \frac{z_j(t) R_j'(t) e^{i \theta_j}}{R_j(t) e^{i \theta_j}-\zeta}  dt .  \eqno(9)
%$$
%in order not to deal with the functions $t_s^{\pm}(\theta)$  and not to integrate alond the different borders of the same slit.  
\end{proof}

If the given contour $L_s$, $s\in \{1, \ldots, n\}$, possesses a cusp then the flow is steady if its branching point for $L_s$ coincides with the cusp.  In order to put the branching points of the flow at the cusps of the given contours so that the velocity there vanishes we have to modify the circulations $\Gamma_k$ around the  contours $L_k$, $k=1, \ldots, n$.

\section*{3. Examples}
1. Consider two elliptic holes with the boundaries $\exp \left(\frac{1}{4} (-i) \pi \right) (3 \exp (i t)-\exp (-i t))+10 i$ and $2 \exp (-i t)+4 \exp (i t)$, $t \in [0, 2 \pi]$, with velocity at infinity equal to $1 + 0.1 i$ (Fig. 1). The circulations $\Gamma_1=\Gamma_2=0$. We give the flow lines on Fig. 1 (a). Fig. 1 (b) shows the absolute values  of velocity. The minimal  velocity values happen at the flow branching points.
  \begin{figure}[h]\label{fig0}
   \begin{center}
   \includegraphics[width=3truecm,height=3truecm]{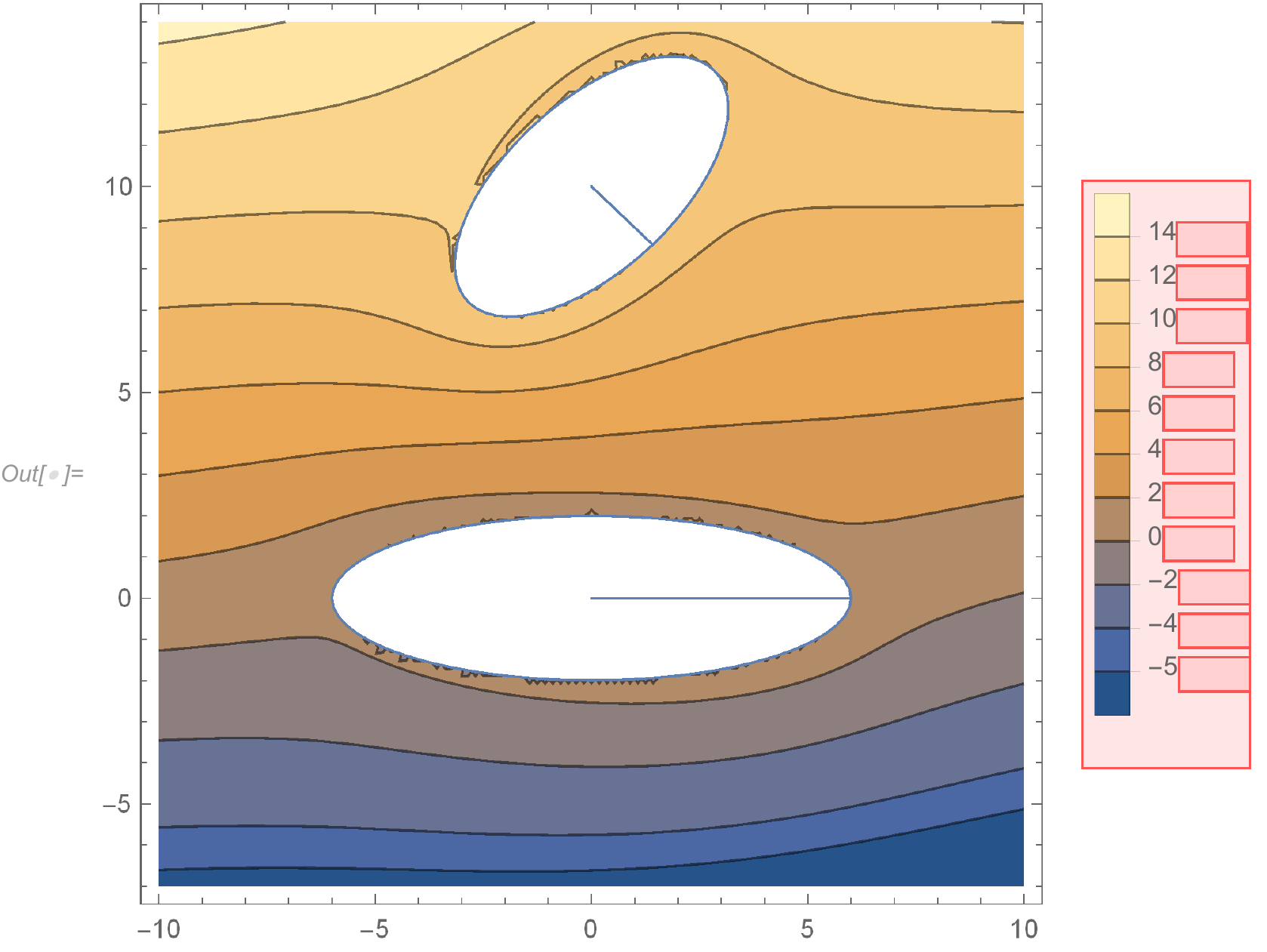}
\includegraphics[width=3truecm,height=3truecm]{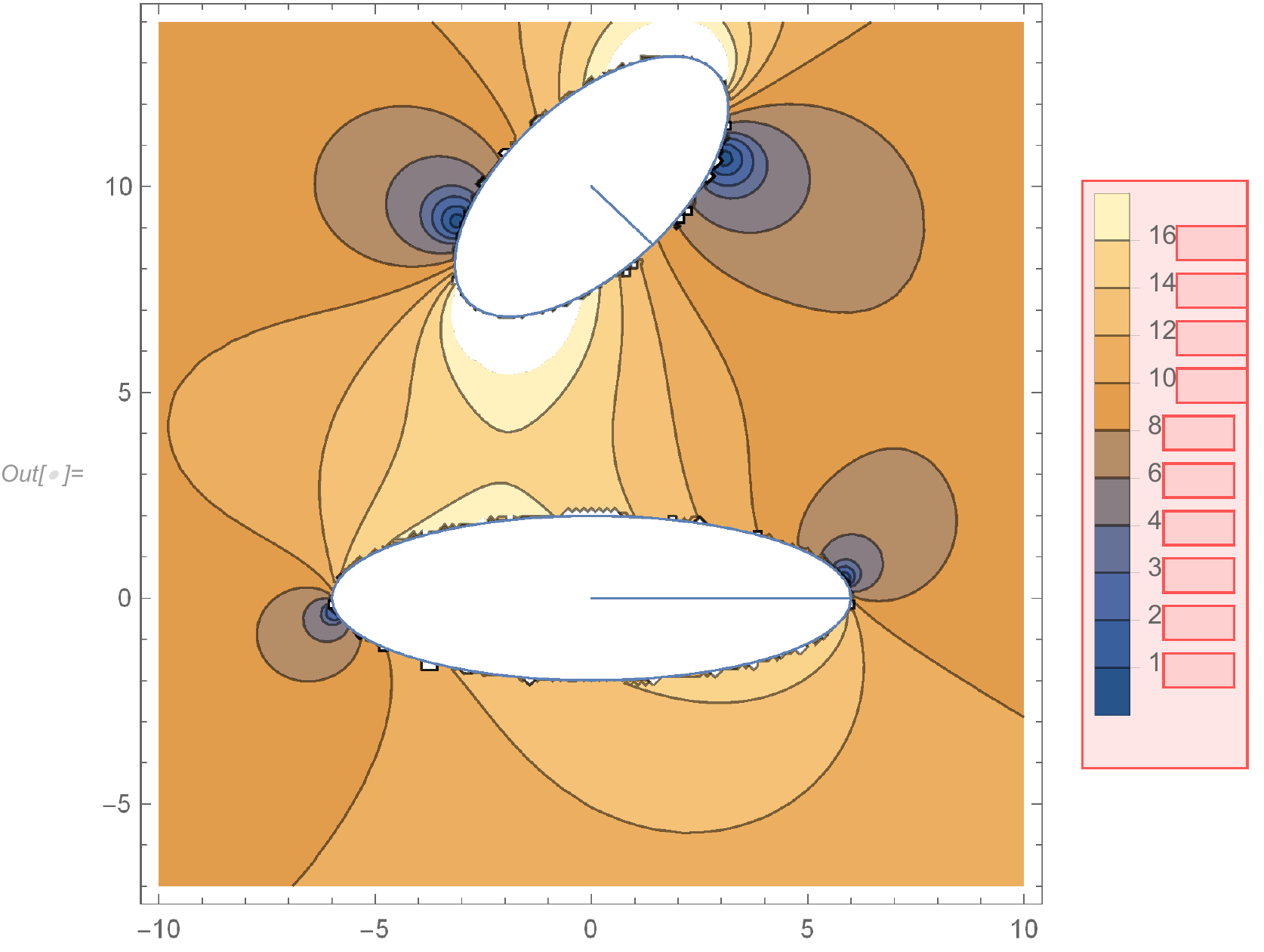}
\caption{The flow around two elliptic holes without circulations. }
  \end{center}
  \end{figure}
  
2. Consider two elliptic holes bounded by $\exp \left(\frac{1}{4} (-i) \pi \right) (3 \exp (i t)-\exp (-i t))+10 i$ and $2 \exp (-i t)+4 \exp (i t)$, $t \in [0, 2 \pi]$, with velocity at infinity equal to $1 + 0.1 i$. In order to shift the branching points of the flow at the contours we introduce the circulations $-0.4 \pi$ and $-1.2 \pi$ around the lower and the upper  contours, respectively.  We give the flow lines on Fig. 2 (a). Fig. 2 (b) shows the absolute values  of velocity.  The flow branching points on the lower contour shift downwards from that of item 1.
  \begin{figure}[h]\label{fig01}
   \begin{center}
   \includegraphics[width=3truecm,height=3truecm]{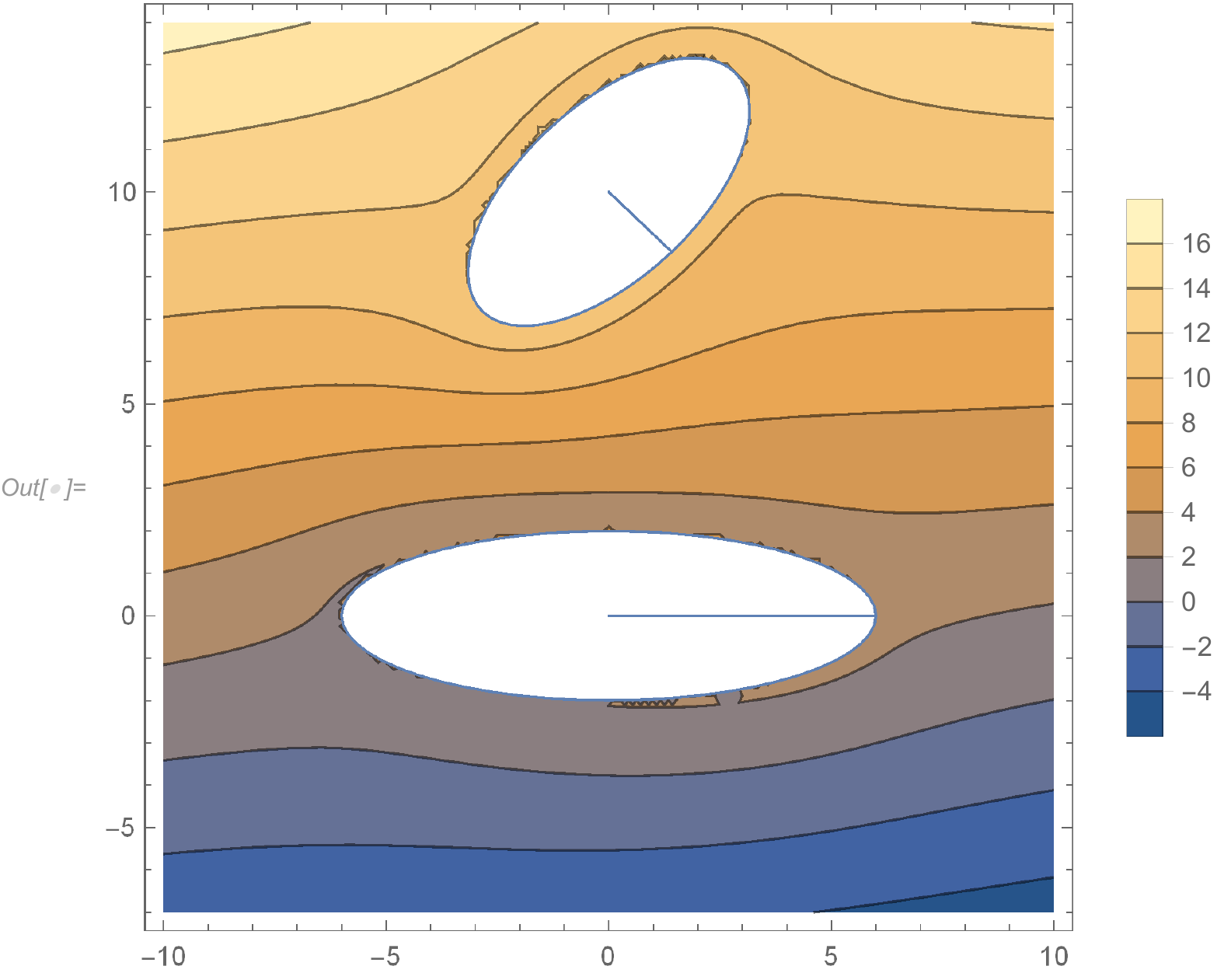}
\includegraphics[width=3truecm,height=3truecm]{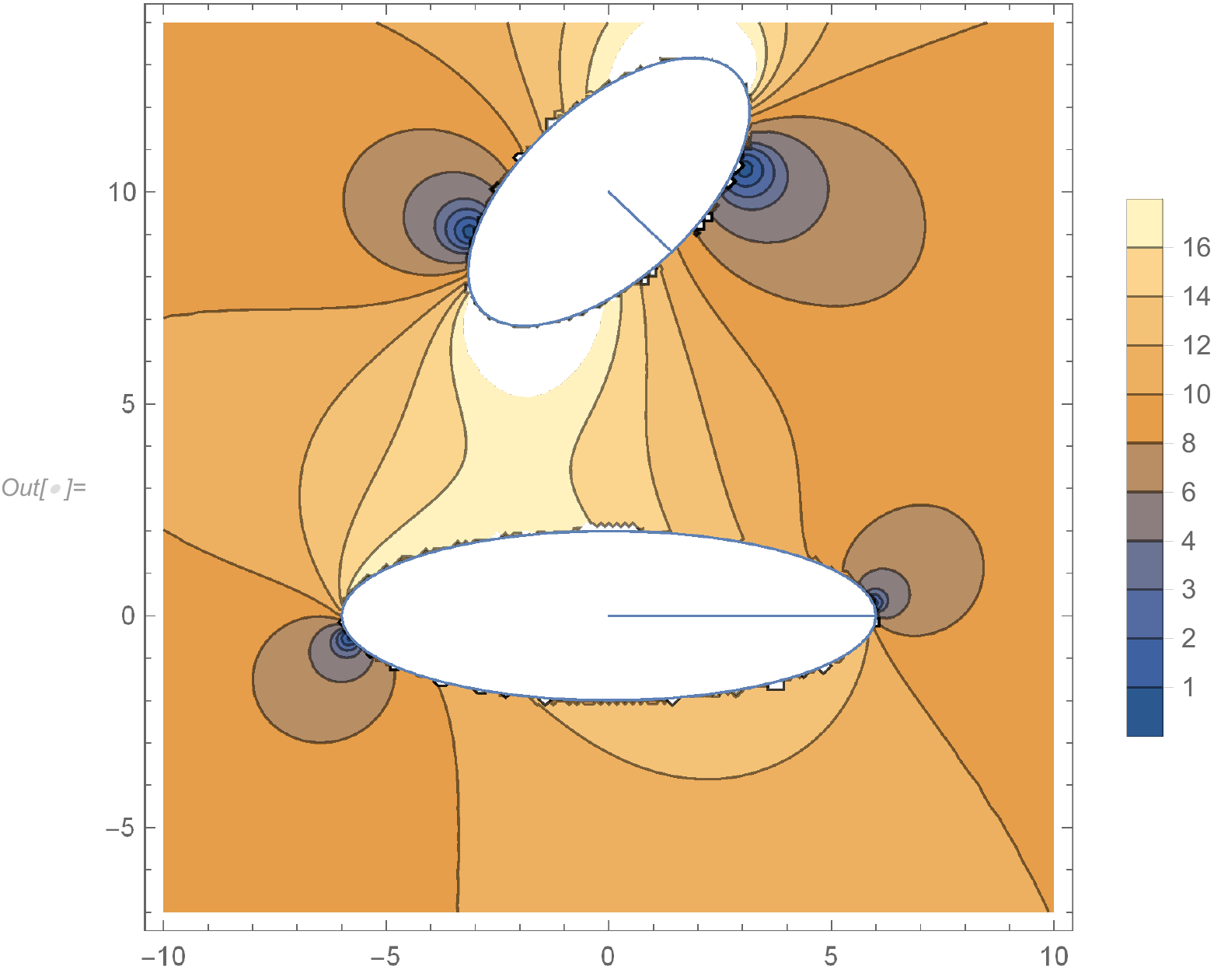}
\caption{The flow around two elliptic holes with  circulations $-0.4 \pi$ and $-1.2 \pi$.}
  \end{center}
  \end{figure}
  
 3. Consider two elliptic holes bounded by $\exp \left(\frac{1}{4} (-i) \pi \right) (3 \exp (i t)-\exp (-i t))+10 i$ and $2 \exp (-i t)+4 \exp (i t)$, $t \in [0, 2 \pi]$, with velocity at infinity equal to $1 + 0.1 i$. In order to shift the branching points of the flow at the contours we introduce the circulations $0.4 \pi$ and $-1.2 \pi$  around the lower and the upper contours, respectively.  We give the flow lines on Fig. 3 (a). Fig. 3 (b) shows the absolute values  of velocity.  The difference in placement of the flow branching points compared to the previous example is clear at the upper contour.
  \begin{figure}[h]\label{fig02}
   \begin{center}
   \includegraphics[width=3truecm,height=3truecm]{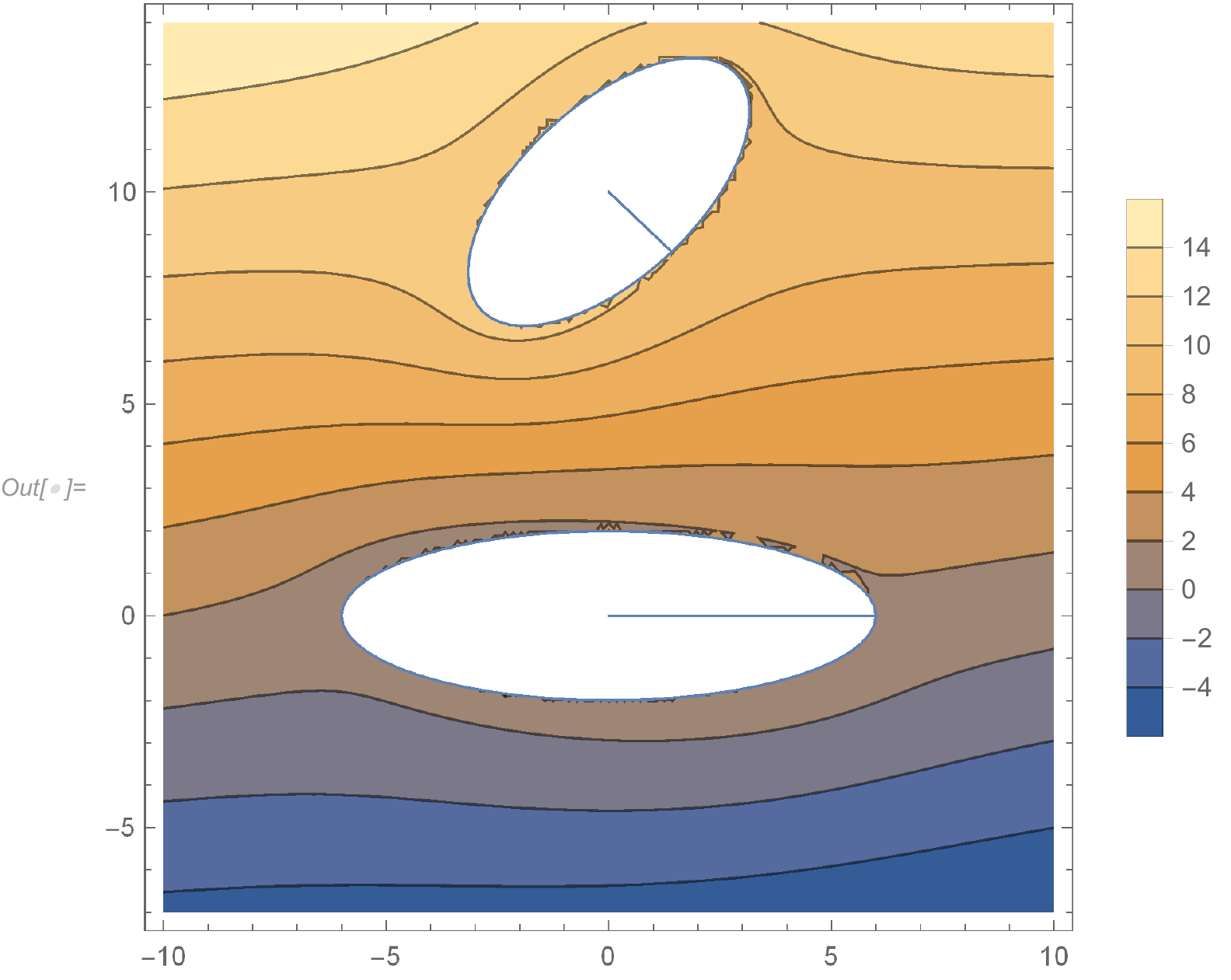}
\includegraphics[width=3truecm,height=3truecm]{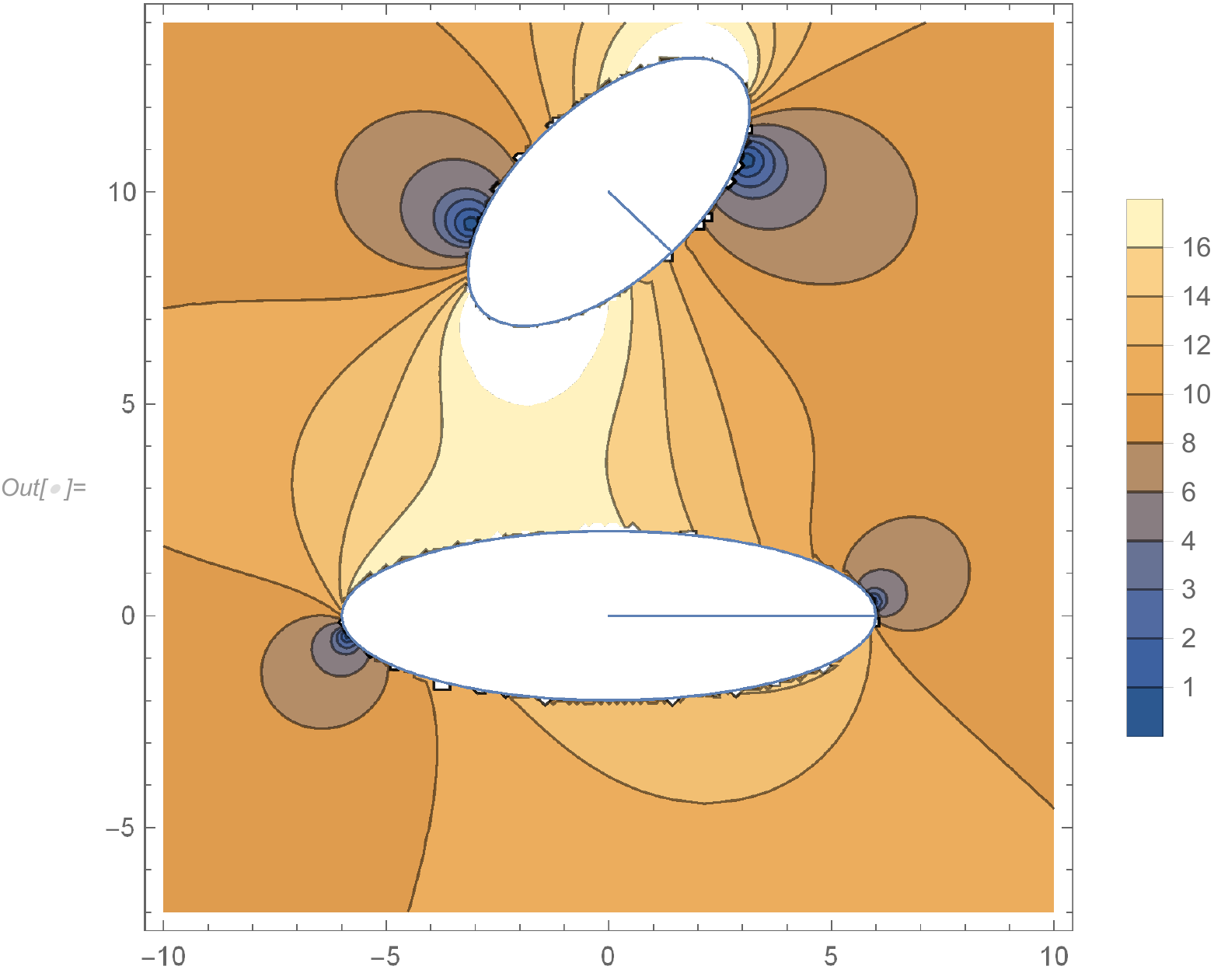}
\caption{The flow around two elliptic holes with  circulations $0.4 \pi$ and $-1.2 \pi$.}
  \end{center}
  \end{figure}
% 1. Flow over the elliptic domain. Consider the elliptic domain bounded by the curve $6 e^{i t}-e^{-i t}$. We construct the flow first without circulation (Fig. 1) and then with  nontrivial circulation (Fig. 2). Note that the figures completely agree with the theoretically predicted picture
%
%  \begin{figure}[h]\label{fig1}
%   \begin{center}
%\includegraphics[width=3truecm,height=3truecm]{fl12.pdf}   
%\includegraphics[width=3truecm,height=4truecm]{fl3.pdf}
%\caption{Flow over the elliptic domain }
%  \end{center}
%  \end{figure}
%
%  
% 2. Noncircular domain with the cusp point. We construct the flow first without circulation (Fig. 3) %and then with  nontrivial circulation (Fig. 4). 
%
%  \begin{figure}[h]\label{fig1}
%   \begin{center}
%   \includegraphics[width=3truecm,height=3truecm]{fl1.pdf}
%%\includegraphics[width=3truecm,height=4truecm]{gr17h.pdf}
%\caption{Flow over the domain with the cusp}
%  \end{center}
%  \end{figure}

4. Consider two noncircular  domains with the cusp points $2\left( 0.15 e^{2 i t}+1.6 e^{i t}+\left( 0.3 i+0.8\right)  e^{-i t}-0.25 i {e}^{-2 i t}\right)$ and  $ 0.15 e^{2 i t}+1.6 e^{i t}+\left( 0.3 i+0.8\right)  e^{-i t}-0.25 i {e}^{-2 i t} +10 i$, $t \in [0, 2 \pi]$. We construct the flow with velocity at infinity equal to $1 + 0.1 i$.   We give the flow lines on Fig. 4 (a). Fig. 4 (b) shows the absolute values  of velocity. 
    \begin{figure}[h]\label{fig1}
   \begin{center}
   \includegraphics[width=3truecm,height=3truecm]{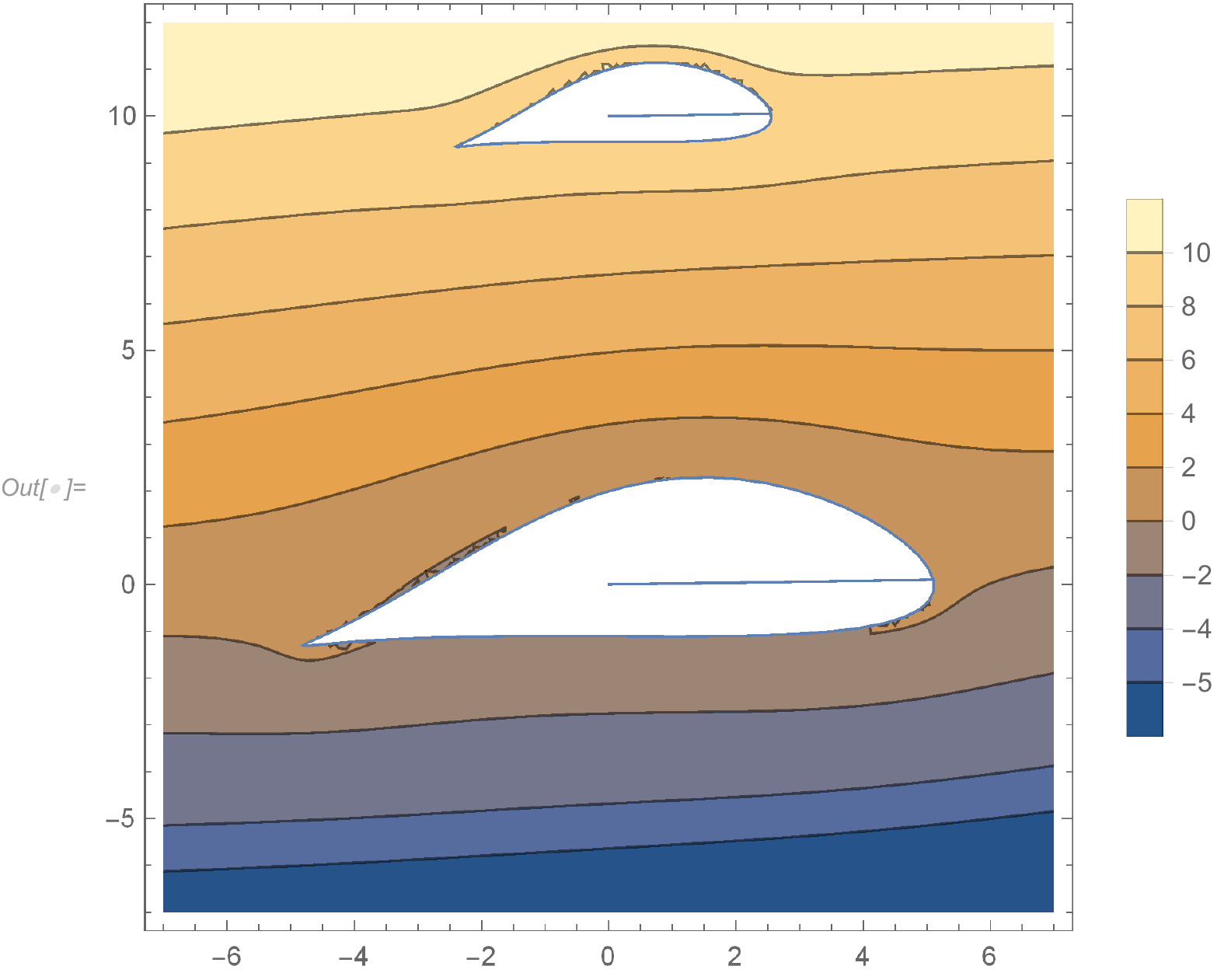}
\includegraphics[width=3truecm,height=3truecm]{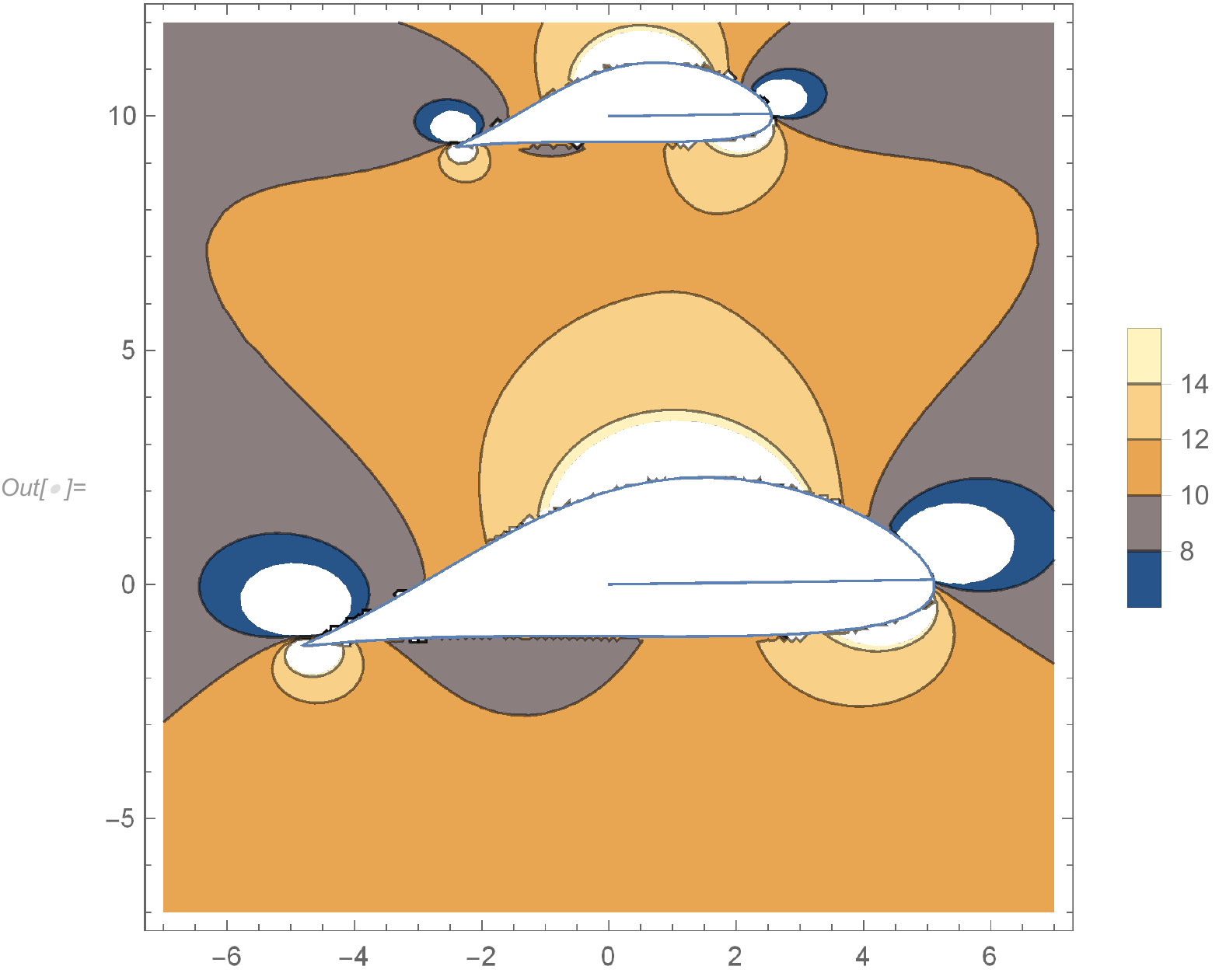}
\caption{Flow  around two domains  with cusps.}
  \end{center}
  \end{figure}

5. Consider two noncircular domains with the cusp points $2\left( 0.15 e^{2 i t}+1.6 e^{i t}+\left( 0.3 i+0.8\right)  e^{-i t}-0.25 i {e}^{-2 i t}\right)$ and  $\left( \frac{i}{2}+\frac{\sqrt{3}}{2}\right)  \left( 0.15 e^{2 i t}+1.6 e^{i t}+\left( 0.3 i+0.8\right)  e^{-i t}-0.25 i {e}^{-2 i t}\right) +10 i$, $t \in [0, 2 \pi]$. We put the flow velocity at infinity equal to $1 + 0.1 i$.  We give the flow lines on Fig. 5 (a). Fig. 5 (b) shows the absolute values  of velocity. %In order to shift the flow branching points we introdce the circulations $-0.9 \pi$ and $-1.4 \pi$ around the lower and the upper contours, respectively.

      \begin{figure}[h]\label{fig2}
   \begin{center}
   \includegraphics[width=3truecm,height=3truecm]{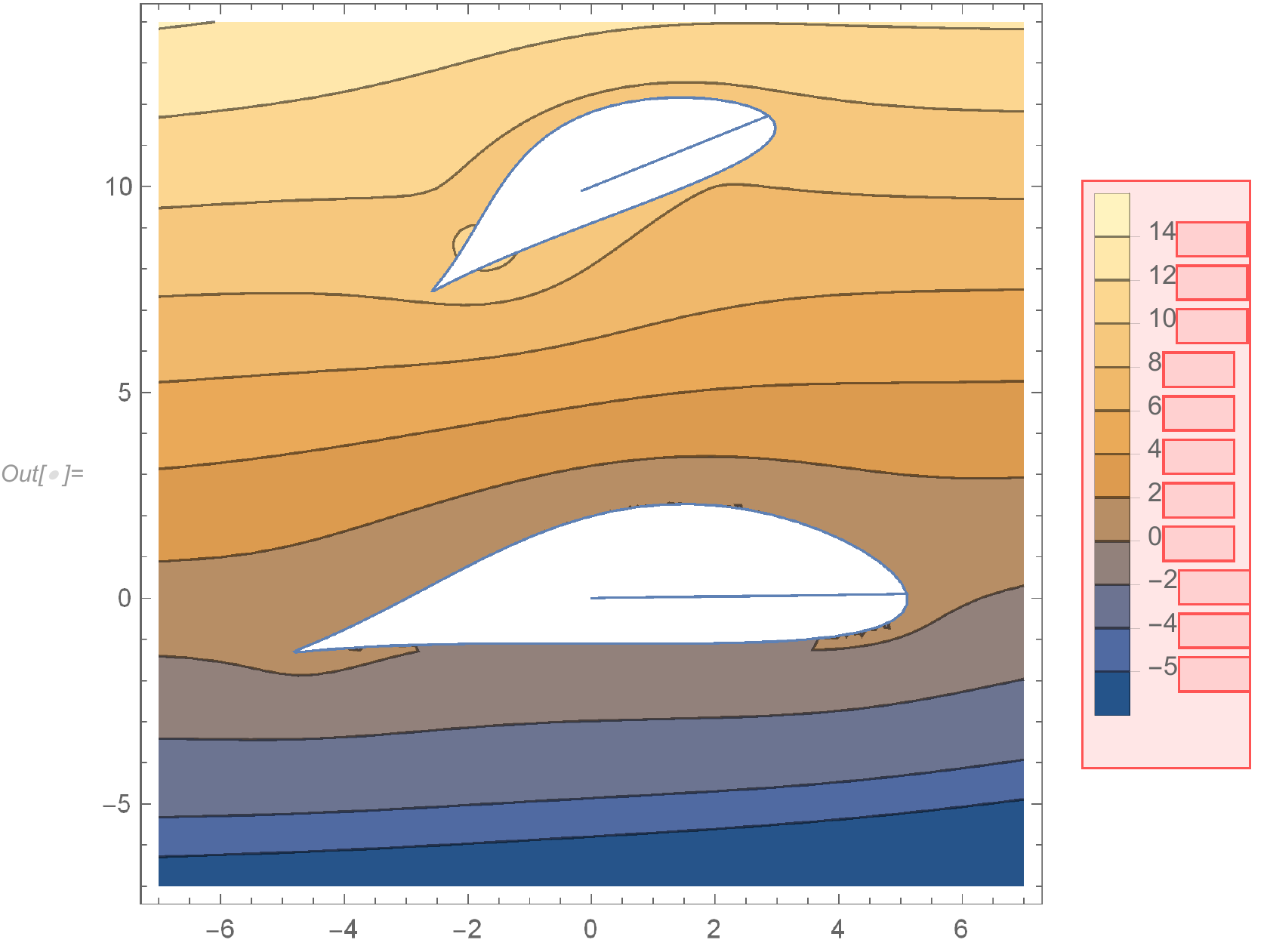}
\includegraphics[width=3truecm,height=3truecm]{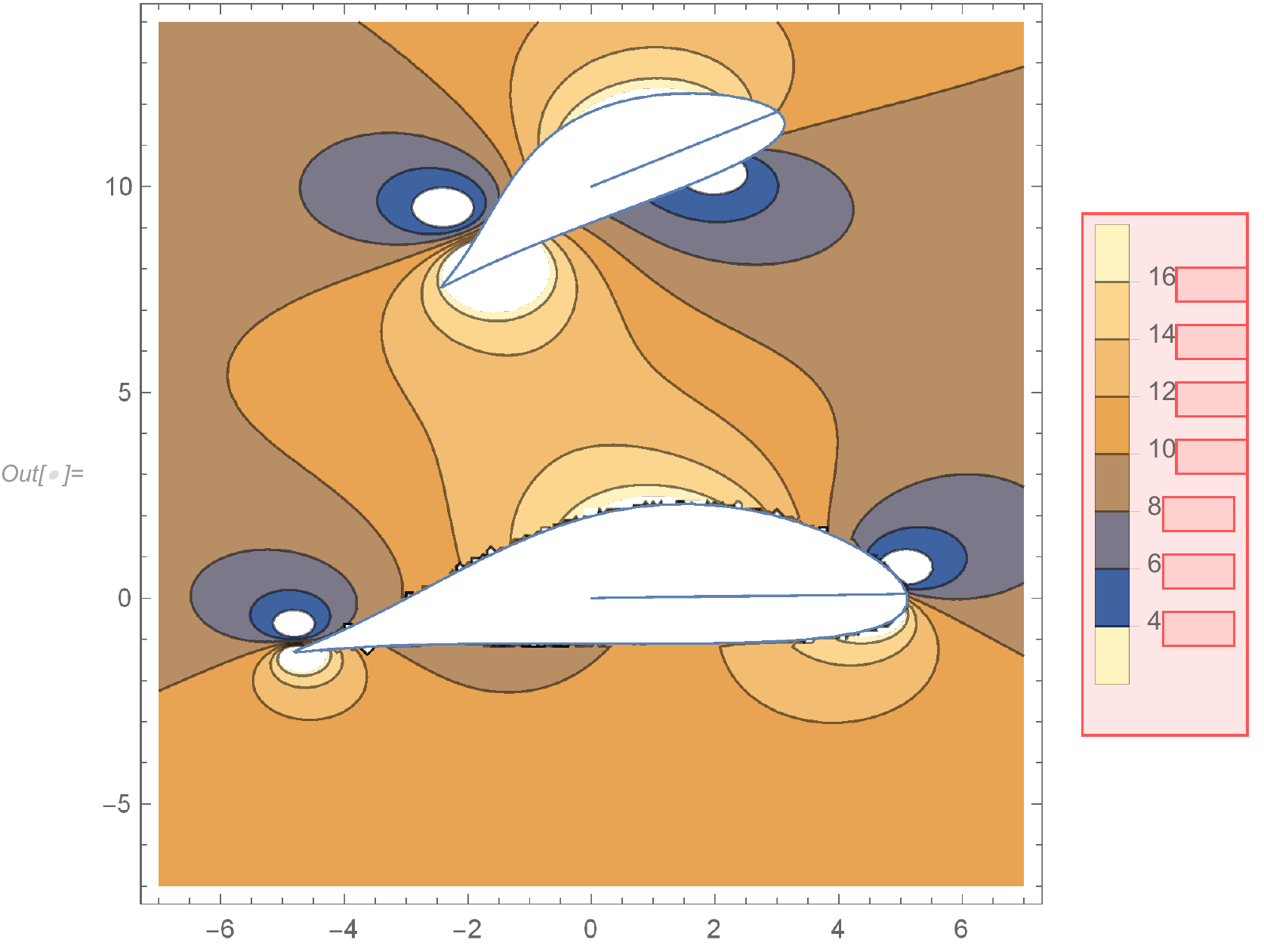}
\caption{Flow  around two domains  with cusps.}
  \end{center}
  \end{figure}
%      \begin{figure}[h]\label{fig1}
%   \begin{center}
%   \includegraphics[width=3truecm,height=3truecm]{out0112.pdf}
%\includegraphics[width=3truecm,height=3truecm]{out01112.pdf}
%\caption{Two domains with cusps}
%  \end{center}
%  \end{figure} 

Consider different possible mutual positioning of the holes from Example 5: $2\left( 0.15 e^{2 i t}+1.6 e^{i t}+\left( 0.3 i+0.8\right)  e^{-i t}-0.25 i {e}^{-2 i t}\right)$ and  $\left( \frac{i}{2}+\frac{\sqrt{3}}{2}\right)  \left( 0.15 e^{2 i t}+1.6 e^{i t}+\left( 0.3 i+0.8\right)  e^{-i t}-0.25 i {e}^{-2 i t}\right) +5 i$. Again  we put the flow velocity at infinity equal to $1 + 0.1 i$. We give the flow lines on Fig. 6 (a). Fig. 6 (b) shows the absolute values  of velocity. Note the large high velocity zone between the contours. % In order to shift the flow branching points we introdce the circulations $-0.6 \pi$ and $-0.4 \pi$ around the lower and the upper contours, respectively. 
       \begin{figure}[h]\label{fig3}
   \begin{center}
   \includegraphics[width=3truecm,height=2.4truecm]{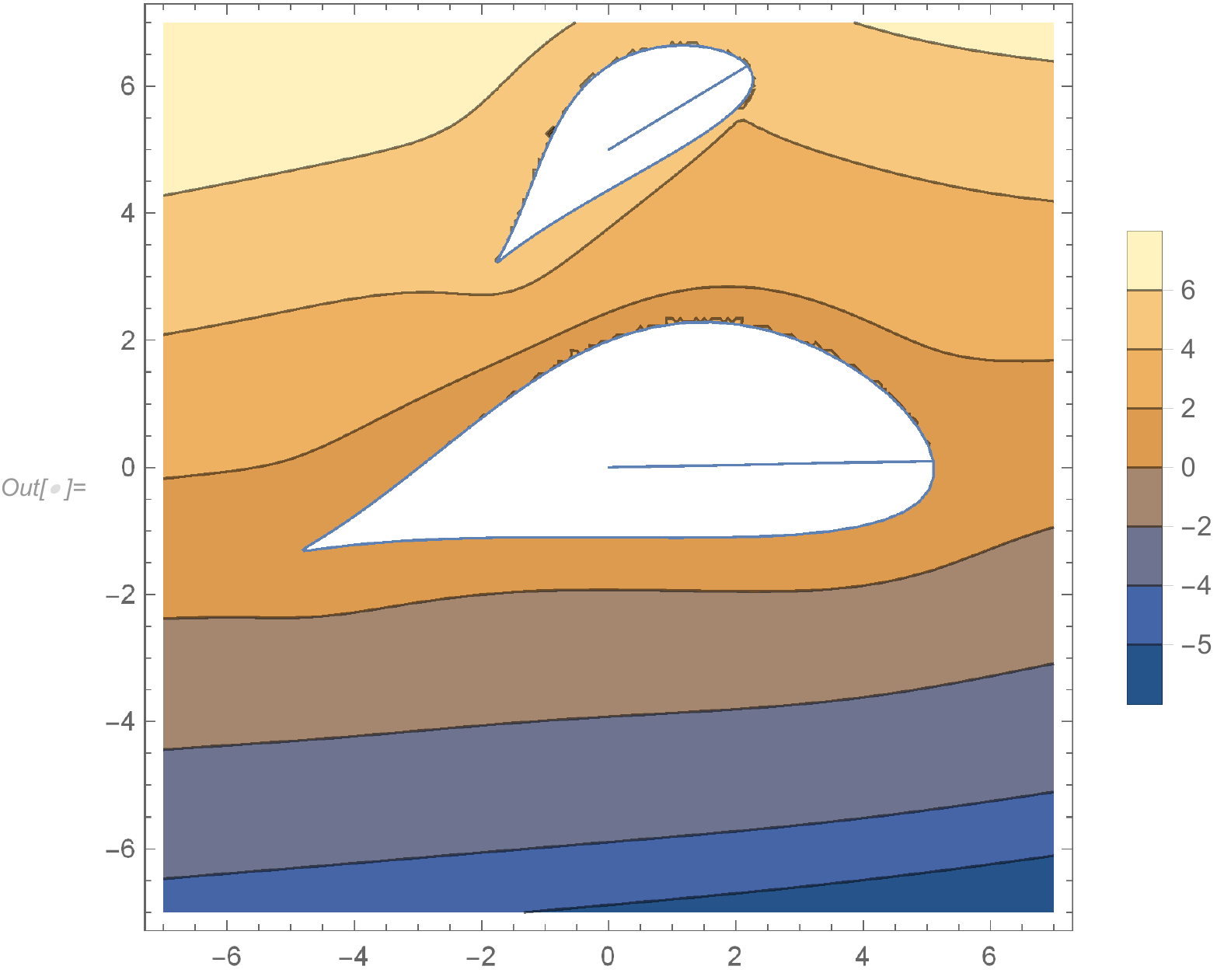}
\includegraphics[width=3truecm,height=2.4truecm]{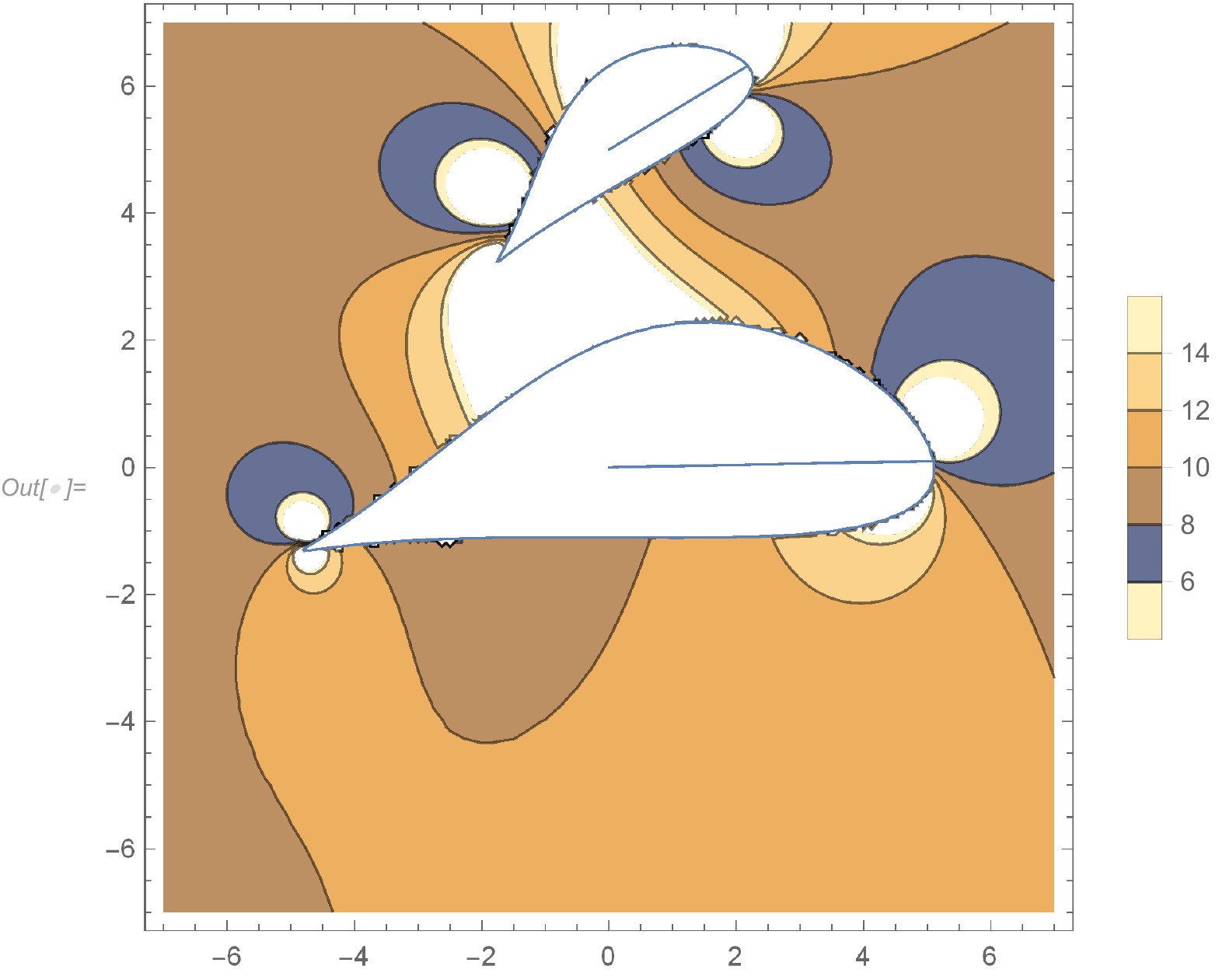}
\caption{Flow around two domains  with cusps.}
  \end{center}
  \end{figure}

\end{document}